\documentclass[letterpaper, 10 pt, conference]{ieeeconf}
\IEEEoverridecommandlockouts
\usepackage[T1]{fontenc}
\usepackage[utf8]{inputenc}
\usepackage{bm}
\usepackage[cmex10]{amsmath}
\usepackage{amssymb}
\usepackage{amsfonts}
\usepackage{mathtools}
\usepackage{graphicx}
\usepackage{booktabs}
\usepackage{cite}
\usepackage{url}
\usepackage{algorithm}
\usepackage{algorithmic}
\usepackage{placeins}
\usepackage{flushend}
\usepackage{xcolor}
\graphicspath{{./figures/}}

\newcommand{\rev}[1]{#1}
\newenvironment{revision}{}{}

\usepackage{amsthm}
\newtheorem{theorem}{Theorem}
\newtheorem{lemma}{Lemma}
\newtheorem{assumption}{Assumption}
\newtheorem{proposition}{Proposition}
\newtheorem{remark}{Remark}

\newcommand{\RR}{\mathbb{R}}
\newcommand{\EE}{\mathbb{E}}
\newcommand{\one}{\bm{1}}
\newcommand{\bsx}{\bm{x}}

\newcommand{\bsL}{\bm{L}}
\newcommand{\bsK}{\bm{K}}
\newcommand{\bsH}{\bm{H}}
\newcommand{\col}{\operatorname{col}}
\newcommand{\sgn}{\operatorname{sgn}}
\newcommand{\norm}[1]{\left\|#1\right\|}
\newcommand{\pbmap}[1]{\mathcal{P}_{\gamma,\tau}\!\left(#1\right)}
\newcommand{\amap}[1]{\mathcal{A}_{\gamma,\tau,\beta_k}\!\left(#1\right)}

\begin{document}

\title{\LARGE \bf
Nonlinear-Gain Distributed Zeroth-Order Optimization\\
for Networked Black-Box Control\\
\large Extended Version}

\author{Shengjun~Zhang, Tingyi~Liu, Heng~Zhang, and Dong~Xie%
\thanks{This work was partially supported by the National Natural Science Foundation of China under grant 52307120 and the National Social Science Fund of China under grant 24CTJ010.}
\thanks{Shengjun~Zhang is with the School of Artificial Intelligence, Hubei University, Wuhan, China. \tt\small{sj.zhang@hubu.edu.cn}.}
\thanks{Tingyi~Liu is with the School of Economics and Management, Wuhan University, Wuhan, China. \tt\small{tingyiL@whu.edu.cn}.}
\thanks{Heng~Zhang is with the School of Electrical Engineering, Shanghai Jiao Tong University, Shanghai, China. \tt\small{zhangheng\_sjtu@sjtu.edu.cn}.}
\thanks{Dong~Xie is with Baidu, Beijing, China. {\tt\small xiedong04@baidu.com}.}}

\maketitle
\thispagestyle{empty}
\pagestyle{empty}

\begin{abstract}
\begin{revision}
This letter studies distributed stochastic optimization over a peer-to-peer network when agents can query only zeroth-order function values. We propose ZOOM-PB, a coordinate-sampling method that blends each local ZO estimate with a fractional-power response while maintaining only a primal state. The raw estimate is retained as a linear anchor, and the nonlinear mixing weight is coupled to the optimization stepsize. This design is motivated by a basic obstruction: transforming heterogeneous or noisy local estimates before averaging can reverse the network direction. We bound that nonlinear residual directly from the raw oracle assumptions instead of imposing an aggregate-alignment condition. With a smooth stochastic-function oracle and a connected graph, ZOOM-PB attains the nonconvex stationarity order $\mathcal{O}(\sqrt{p/(nT)})$ and a Polyak--{\L}ojasiewicz statistical term of order $\mathcal{O}(p/(nT))$, after an explicit initialization transient. Numerical examples compare ZOOM-PB with seven distributed ZO baselines under matched query and message budgets.
\end{revision}
\end{abstract}


\section{Introduction}

Distributed optimization is a basic tool for networked control, machine learning, robotic swarms, and sensor networks, where agents cooperatively solve
\begin{equation}\label{eq:prob}
    \min_{x\in\RR^p} f(x)\triangleq \frac{1}{n}\sum_{i=1}^{n} f_i(x),
    \qquad f_i(x)=\EE_{\xi_i}[F_i(x,\xi_i)] .
\end{equation}
Many first-order distributed methods assume that each agent can compute or estimate $\nabla f_i(x)$ directly. This assumption is restrictive in black-box control and simulation-based optimization, where agents may observe only noisy function values, such as concentration readings in source seeking~\cite{li2020cooperative} or loss values in black-box learning~\cite{spall2005introduction,conn2009introduction,audet2017derivative,larson2019derivative}.

Zeroth-order (ZO) methods replace gradients by finite-difference estimates and have been extensively studied in centralized stochastic optimization~\cite{ghadimi2013stochastic,lian2016Comprehensive,nesterov2017random,liu2018zeroth}. Distributed ZO methods further couple gradient estimation with network consensus~\cite{yuan2014randomized,yuan2015gradient,tang2020distributedzero,hajinezhad2019zone,yi2021zerothorder,zhang2021convergence,tang2023zerothfeedback}. However, most existing distributed stochastic ZO algorithms follow one of two paths: they either use spherical perturbations with primal--dual or tracking variables, or they focus on asymptotic stationarity without addressing the transient slowdowns caused by small ZO signals in flat nonconvex regions. \rev{This issue arises naturally in black-box networked control, where useful finite-difference signals may be small relative to measurement or sampling noise.}

\begin{revision}
We introduce a scaled \emph{powerball} map into distributed ZO updates. For an estimated component $g$ and a reference scale $\tau>0$, define
\begin{equation}\label{eq:pb}
    \mathcal{P}_{\gamma,\tau}(g)
    =\tau^{1-\gamma}\sgn(g)|g|^\gamma,
    \qquad \gamma\in[1/2,1],
\end{equation}
applied componentwise. Its gain ratio is
\begin{equation}
    \frac{|\mathcal{P}_{\gamma,\tau}(g)|}{|g|}
    =\left(\frac{\tau}{|g|}\right)^{1-\gamma}.
\end{equation}
Hence components below $\tau$ are amplified and components above $\tau$ are attenuated. The exact scaling laws are
$\mathcal{P}_{\gamma,\tau}(cg)=c^\gamma\mathcal{P}_{\gamma,\tau}(g)$ and
$\mathcal{P}_{\gamma,c\tau}(cg)=c\mathcal{P}_{\gamma,\tau}(g)$ for $c>0$.
The map reshapes both signal and noise; it is not a noise discriminator.
Powerball maps were introduced for first-order descent~\cite{yuan2019powerball}; here the map is placed after a local coordinate-ZO oracle and before peer-to-peer averaging.

To retain a reliable linear component, ZOOM-PB uses the anchored map
\begin{equation}\label{eq:anchored_map}
 \mathcal A_{\gamma,\tau,\beta}(g)
 =(1-\beta)g+\beta\mathcal P_{\gamma,\tau}(g),
 \qquad \beta\in[0,1].
\end{equation}
The map is evaluated locally; $\beta=0$ gives the raw ZO direction and
$\beta=1$ gives the pure powerball direction.

The paper makes three contributions. First, ZOOM-PB combines sampled coordinate queries with a one-state recursion: each agent maintains and broadcasts only $x_{i,k}$. Local oracle calls run in parallel, but each agent still queries only its own objective. Second, we show that the nonlinear residual in~\eqref{eq:anchored_map} is a controlled perturbation under the same raw stochastic-oracle assumptions used for the linear direction. Choosing $\beta_k^2=\mathcal O(\eta_k/n)$ preserves the stated nonconvex and PL statistical orders without assuming that the averaged pure-powerball direction is aligned. The proof also bounds the nonlinear forcing in the consensus recursion. Third, matched-budget experiments separate gain shaping from estimator and state choices across coordinate, spherical, and Gaussian distributed-ZO baselines.
\end{revision}

\begin{proposition}\label{prop:gain}
\rev{For $\gamma\in(0,1)$, $\tau>0$, $\beta\in(0,1]$, and $g\ne0$,
$\mathcal A_{\gamma,\tau,\beta}$ preserves the sign of $g$, and
\[
 \frac{|\mathcal A_{\gamma,\tau,\beta}(g)|}{|g|}
 =(1-\beta)+\beta\left(\frac{\tau}{|g|}\right)^{1-\gamma}.
\]
It amplifies components below $\tau$, leaves $|g|=\tau$ unchanged, and
attenuates components above $\tau$. Moreover,
$\mathcal A_{1,\tau,\beta}(g)=g$.}
\end{proposition}
\begin{proof}
\rev{Both terms in~\eqref{eq:anchored_map} have the sign of $g$. Dividing
their convex combination by $|g|$ gives the displayed ratio.}
\end{proof}

\rev{Proposition~\ref{prop:gain} describes each scalar input, not the direction obtained after averaging across agents or noise. Lemma~\ref{lem:key} therefore analyzes the difference $\mathcal P_{\gamma,\tau}(g^e)-g^e$ without assigning it a descent sign. The gain does not improve the claimed asymptotic exponent; its role is finite-time input shaping.}

\begin{revision}
\begin{table*}[t]
\caption{Representative ZO methods: local state, sent vector, and rates in the iteration/communication-round index of each cited result.}
\label{tab:position}
\centering
\scriptsize
\setlength{\tabcolsep}{2.2pt}
\begin{tabular}{@{}p{0.19\textwidth}p{0.17\textwidth}p{0.15\textwidth}p{0.13\textwidth}p{0.25\textwidth}@{}}
\toprule
Method & Oracle & Local & Sent & Guarantee or limitation \\
\midrule
Centralized ZO~\cite{ghadimi2013stochastic,lian2016Comprehensive} & Stochastic & server $x$ & -- & $\mathcal{O}(\sqrt{p/T})$ \\
ZO-GDA~\cite{tang2020distributedzero} & Spherical 2-point & $x$ & $x$ & Deterministic setting \\
ZONE-M~\cite{hajinezhad2019zone} & Gaussian multi-point & $z^r,z^{r-1},\bar G^{r-1}$ & $z^r$ & $\mathcal{O}(1/T)$, $J=\mathcal O(T)$ \\
ZOD-PDA/ZOD-PA~\cite{yi2021zerothorder} & Spherical 2-point & $(x,v)$/$x$ & $x$ & $\mathcal O(\sqrt{p/(nT)})$; PL $\mathcal O(p/(nT))$ \\
{\raggedright ZODIAC/accelerated powerball ZO~\cite{zhang2021convergence,zhang2022accelerated}\par} & Sampled coordinate & $(x,v)$ & $x$ & $\mathcal O(\sqrt{p/T})$/$\mathcal O(\sqrt{p/(nT)})$ \\
ZOOM-PB & Sampled coordinate & $x$ & $x$ & $\mathcal{O}(\sqrt{p/(nT)})$; PL $\mathcal{O}(p/(nT))$ \\
\bottomrule
\end{tabular}
\end{table*}
\end{revision}

\subsection{Relation to Existing Distributed ZO Methods}

\begin{revision}
Centralized stochastic ZO establishes the baseline dimension dependence but does not address peer-to-peer agreement. Distributed deterministic ZO treats exact local costs, whereas the stochastic-function model below represents data sampling or simulator randomness using the same sample for paired evaluations. Independent additive sensor noise is distinguished in Remark~\ref{rem:indnoise}.

The closest coordinate methods are ZODIAC and the ACC 2022 powerball method~\cite{zhang2022accelerated}. Both maintain $(x_{i,k},v_{i,k})$, but their published protocols broadcast only $x_{i,k}$. ZOOM-PB uses the same one-vector message while eliminating the local dual recursion and anchoring the nonlinear response. Thus its distinction is state and update structure, not a smaller link payload relative to those methods; powerball acceleration itself is not claimed as new. ZOD-PDA and ZOD-PA~\cite{yi2021zerothorder} attain the same leading stochastic nonconvex and PL orders using spherical probes. ZOOM-PB removes ZOD-PDA's local dual vector and its update. ZOD-PA is already primal-only, so the comparison there is not a state, message, or rate advantage: ZOOM-PB instead offers sparse sampled-coordinate probes and an anchored nonlinear interpolation whose finite-time effect is tested below. Multi-agent ZO feedback optimization provides a complementary control-loop perspective~\cite{tang2023zerothfeedback}.

For $\gamma<1$, raw unbiasedness and bounded heterogeneity do not imply that
$n^{-1}\sum_i\mathcal{P}_{\gamma,\tau}(g^e_{i,k})$ aligns with the global gradient; local gradients or asymmetric noise can reverse its sign. Equation~\eqref{eq:anchored_map} does not assume away this effect. It retains the raw direction and treats the nonlinear difference as a perturbation whose weight decreases with the stepsize. Thus the rate result covers the implemented anchored update rather than a hypothetical interchange of transformation and averaging.
\end{revision}

\section{Problem Setup and Algorithm}

Let $\mathcal{G}=(\mathcal{V},\mathcal{E})$ be an undirected connected graph with Laplacian $L$. Define $\bsL=L\otimes I_p$, $K_n=I_n-\one\one^\top/n$, $\bsK=K_n\otimes I_p$, $\bsH=(\one\one^\top/n)\otimes I_p$, and $\bar{x}_k=\frac{1}{n}\sum_{i=1}^n x_{i,k}$. The agents minimize~\eqref{eq:prob} using only evaluations of $F_i(\cdot,\xi_i)$.

\begin{assumption}\label{ass:regularity}\label{ass:graph}\label{ass:optset}
The undirected graph $\mathcal{G}$ is connected. The optimal set is nonempty, and the optimal value satisfies $f^*>-\infty$.
\end{assumption}

\begin{assumption}\label{ass:smooth}
\rev{For almost every $\xi_i$, the sample function $F_i(\cdot,\xi_i)$ is $L_f$-smooth.}
\end{assumption}

\begin{assumption}\label{ass:oracle}\label{ass:variance}\label{ass:sampling}
\rev{For every coordinate $j$,
$\EE[(\nabla F_i(x,\xi_i)-\nabla f_i(x))_j^2]\le \zeta^2$, and
$\norm{\nabla f_i(x)-\nabla f(x)}^2\le\varsigma^2$ for all $i,x$.
Conditioned on the history, $\xi_{i,k}$ and the uniformly sampled cardinality-$n_c$ sets $\mathcal S_{i,k}$ are independent across agents and iterations. The same sample $\xi_{i,k}$ is used in all function values forming one finite-difference estimate.}
\end{assumption}

\begin{revision}
Let $e_\ell\in\RR^p$ denote the $\ell$th canonical unit vector. At iteration $k$, agent $i$ samples $\mathcal{S}_{i,k}\subseteq\{1,\ldots,p\}$ and constructs either
\begin{align}
g^e_{i,k}
&=\frac{p}{n_c}\sum_{\ell\in\mathcal{S}_{i,k}}
\frac{F_i(x_{i,k}+\delta_{i,k}e_\ell,\xi_{i,k})-F_i(x_{i,k},\xi_{i,k})}
{\delta_{i,k}}e_\ell,\label{eq:est1}\\[-1mm]
g^e_{i,k}
&=\frac{p}{2n_c\delta_{i,k}}\sum_{\ell\in\mathcal{S}_{i,k}}
\Bigl[F_i(x_{i,k}+\delta_{i,k}e_\ell,\xi_{i,k}) \notag\\[-1mm]
&\hspace{28mm}-F_i(x_{i,k}-\delta_{i,k}e_\ell,\xi_{i,k})\Bigr]e_\ell.
\label{eq:est2}
\end{align}
Each agent queries only $F_i$; agents do not share another agent's local sampling burden. The local calls are executed in parallel.
\end{revision}

\begin{remark}\label{rem:indnoise}
\rev{Assumptions~\ref{ass:smooth}--\ref{ass:oracle} specify a stochastic-function/common-random-number oracle. If separate calls return $f_i(x)+w$ with independent variance $\nu_i^2$, the estimator second moment has an additional term of order $p^2\nu_i^2/(n_c\delta_{i,k}^2)$. The theorems below do not silently absorb this term; independent-noise implementations require averaging, a nonvanishing radius, or a modified rate.}
\end{remark}

\begin{algorithm}[t]
\caption{ZOOM-PB}
\label{alg:zoompb}
\begin{algorithmic}[1]
\STATE \textbf{Input}: $\alpha>0$, $\gamma\in[1/2,1]$, $\rev{\tau>0}$, steps $\{\eta_k\}$, weights $\{\beta_k\}\subset[0,1]$, radii $\{\delta_{i,k}\}$.
\STATE \textbf{Initialize}: $x_{i,0}\in\RR^p$ for all $i\in[n]$.
\FOR{$k=0,1,\ldots$}
\FOR{each agent $i$ in parallel}
\STATE Exchange $x_{i,k}$ with neighbors $j\in\mathcal{N}_i$.
\STATE Form $g^e_{i,k}$ by~\eqref{eq:est1} or~\eqref{eq:est2}.
\STATE Set $s_{i,k}=\rev{\amap{g^e_{i,k}}}$.
\STATE Update
\begin{equation*}
x_{i,k+1}=x_{i,k}-\alpha\sum_{j=1}^{n}L_{ij}x_{j,k}
-\eta_k\rev{s_{i,k}}.
\end{equation*}
\ENDFOR
\ENDFOR
\end{algorithmic}
\end{algorithm}

\rev{When $\gamma=1$ or $\beta_k=0$, $s_{i,k}=g^e_{i,k}$ and
Algorithm~\ref{alg:zoompb} reduces to ZOOM, the raw-ZO ablation without
powerball shaping. Computing $s_{i,k}$ is local:
it introduces neither an auxiliary vector recursion nor another message.}

\begin{revision}
Let $\mathcal F_k$ be the history before the oracle calls at time $k$, and define
\begin{align}
r_{i,k}&=\pbmap{g^e_{i,k}}-g^e_{i,k},\qquad
s_{i,k}=g^e_{i,k}+\beta_k r_{i,k}.
\label{eq:anchored_direction}
\end{align}
$\bar g^e_k=n^{-1}\sum_i g^e_{i,k}$,
$\bar r_k=n^{-1}\sum_i r_{i,k}$,
$\chi_k=\norm{\bsK\bsx_k}^2/n$, and
$\delta_k=\max_i\delta_{i,k}$.

The analysis never assigns a favorable sign to $\bar r_k$. Instead,
Lemma~\ref{lem:key} derives its second moment from
Assumptions~\ref{ass:smooth}--\ref{ass:oracle} and uses
$\beta_k^2=\mathcal O(\eta_k/n)$ to absorb the resulting bias. This is the
only additional schedule used for the nonlinear term.
\end{revision}

\begin{revision}
The recursion also makes its resource costs explicit. Per agent and round,
\eqref{eq:est1} uses $n_c+1$ function values and~\eqref{eq:est2} uses
$2n_c$. Over $T$ rounds the corresponding network totals are
$nT(n_c+1)$ and $2nTn_c$. Because each agent sends one $p$-vector to every
neighbor, communication costs $2|\mathcal E|Tp$ transmitted scalars. The
nonlinear map and $(\gamma,\tau,\beta_k)$ do not alter these counts. The
rates below assume bounded $p/n_c$; otherwise that ratio remains explicit
in the oracle-moment constants.
\end{revision}

\section{Convergence Results}

\begin{revision}
Let $\bm s_k=\col(s_{1,k},\ldots,s_{n,k})$ and
$\bar s_k=n^{-1}\sum_i s_{i,k}$. The two exact recursions are
\begin{align}
\bar x_{k+1}&=\bar x_k-\eta_k\bar s_k,\label{eq:average_recursion}\\
\bsK\bsx_{k+1}
&=\bsK(I_{np}-\alpha\bsL)\bsx_k-\eta_k\bsK\bm s_k.
\label{eq:disagreement_recursion}
\end{align}
Thus descent depends on the anchored network average $\bar s_k$, while
consensus depends on the distinct quantity $\bsK\bm s_k$.
For $0<\alpha<2/\lambda_n(L)$, set
$\rho_\alpha=\max_{2\le i\le n}|1-\alpha\lambda_i(L)|<1$.
\end{revision}

For compactness, write $\EE_k[\cdot]=\EE[\cdot\mid\mathcal F_k]$.
\begin{lemma}\label{lem:key}
\begin{revision}
Suppose Assumptions~\ref{ass:regularity}--\ref{ass:oracle} hold and
$p/n_c\le\kappa_c$. There are positive constants independent of $k,T,n,p$ such that
\begin{subequations}\label{eq:direction_moments}
\begin{align}
\frac1n\sum_i\EE_k[\norm{g^e_{i,k}}^2+\norm{r_{i,k}}^2]
&\le R_g\norm{\nabla f(\bar x_k)}^2+R_x\chi_k\notag\\
&\quad+R_0p+R_\delta p^2\delta_k^2,\label{eq:residual_moment}\\
\EE_k[\norm{\bar s_k}^2]
&\le B_g\norm{\nabla f(\bar x_k)}^2+B_x\chi_k\notag\\
&\quad+B_0\left(\frac{p}{n}+\beta_k^2p\right)
+B_\delta p^2\delta_k^2,
\label{eq:moment}\\
\EE_k[\norm{\bsK\bm s_k}^2]
&\le nV_g\norm{\nabla f(\bar x_k)}^2+nV_x\chi_k\notag\\
&\quad+nV_0p+nV_\delta p^2\delta_k^2.
\label{eq:centered_moment}
\end{align}
\end{subequations}
For fixed $\kappa_\beta>0$ and $\beta_{\rm d}\in(0,1)$ satisfying
$16R_g\beta_{\rm d}^2\le1$, there are positive
$c_x,c_g,C_0,C_\delta$, and $\bar\eta$, independent of $k,T,p,n$, such that if
\begin{equation}\label{eq:beta_schedule}
 \eta_k\le\frac{\bar\eta}{n},\qquad
 0\le\beta_k\le\bar\beta_k:=
 \min\!\left\{\beta_{\rm d},\kappa_\beta\sqrt{\frac{\eta_k}{n}}\right\},
\end{equation}
then a nonnegative Lyapunov function
$W_k\asymp n(f(\bar x_k)-f^*)+\chi_k$ satisfies
\begin{equation}
\begin{split}
\EE[W_{k+1}]
\le{}& \EE[W_k]-c_x\EE[\chi_k]\\
&-c_g n\eta_k\EE[\norm{\nabla f(\bar x_k)}^2]\\
&+C_0p\eta_k^2+C_\delta np^2\eta_k\delta_k^2.
\end{split}
\label{eq:lyap}
\end{equation}
The constants and $\bar\eta$ may depend on
$(\gamma,\tau,\beta_{\rm d},\kappa_\beta,\kappa_c,L_f,\zeta,\varsigma)$ and
$(1-\rho_\alpha)^{-1}$, but not on $T$. The displayed $n,p$ scaling is
uniform for graph families with $1-\rho_\alpha$ bounded away from zero.
\end{revision}
\end{lemma}

\begin{proof}
\begin{revision}
Let $m_{i,k}=\EE_k[g^e_{i,k}]$. The coordinate-estimator calculation in
Appendix~A gives
\begin{align}
\left\|\frac1n\sum_i m_{i,k}-\nabla f(\bar x_k)\right\|^2
&\le C\chi_k+Cp\delta_k^2,\label{eq:raw_mean_short}\\
\frac1n\sum_i\EE_k\norm{g^e_{i,k}}^2
&\le C\big(\norm{\nabla f(\bar x_k)}^2+\chi_k+p\big)\notag\\
&\quad+Cp^2\delta_k^2.
\label{eq:raw_local_short}
\end{align}
For every scalar $a$,
\begin{equation}\label{eq:scalar_residual}
 |\mathcal P_{\gamma,\tau}(a)-a|^2
 \le C_{\gamma,\tau}(1+a^2),
\end{equation}
because $|a|^{2\gamma}\le1+a^2$. Summation proves
\eqref{eq:residual_moment}. Conditional independence gives the $p/n$
term for $\bar g^e_k$, while Jensen's inequality bounds $\bar r_k$;
using $\bar s_k=\bar g^e_k+\beta_k\bar r_k$ proves~\eqref{eq:moment}.
The projector inequality
$\norm{\bsK\bm s_k}^2\le\sum_i\norm{s_{i,k}}^2$ proves
\eqref{eq:centered_moment}.

Let $g_k=\nabla f(\bar x_k)$. Equation~\eqref{eq:raw_mean_short} and
Young's inequality give
\[
 \left\langle g_k,\frac1n\sum_i m_{i,k}\right\rangle
 \ge\frac34\norm{g_k}^2-C\chi_k-Cp^2\delta_k^2.
\]
Without assigning a sign to $\bar r_k$,
\[
 \beta_k|\langle g_k,\EE_k\bar r_k\rangle|
 \le\frac18\norm{g_k}^2+2\beta_k^2\EE_k\norm{\bar r_k}^2.
\]
By Jensen's inequality and~\eqref{eq:residual_moment}, the second term
contributes at most $2R_g\beta_k^2\norm{g_k}^2\le
\norm{g_k}^2/8$ under~\eqref{eq:beta_schedule}. Hence
\begin{equation}\label{eq:anchored_alignment}
 \langle g_k,\EE_k\bar s_k\rangle
 \ge\frac12\norm{g_k}^2-C\chi_k-Cp^2\delta_k^2-Cp\beta_k^2.
\end{equation}

Using $\rho_\alpha<1$, Young's inequality in
\eqref{eq:disagreement_recursion}, and~\eqref{eq:centered_moment} yields
\begin{align}
\EE[\chi_{k+1}\mid\mathcal F_k]
&\le(1-\ell_\alpha)\chi_k\notag\\
&\quad+C\eta_k^2\big(\norm{g_k}^2+\chi_k+p+p^2\delta_k^2\big)
\label{eq:consensus_step}
\end{align}
for $\ell_\alpha>0$. Smoothness,~\eqref{eq:anchored_alignment}, and
$n\eta_k\beta_k^2p\le\kappa_\beta^2p\eta_k^2$ give the corresponding
objective decrease with remainder
$C n\eta_k\chi_k+C p\eta_k^2+Cnp^2\eta_k\delta_k^2$.
A fixed positive combination with~\eqref{eq:consensus_step} absorbs the
disagreement term whenever $n\eta_k\le\bar\eta$, proving~\eqref{eq:lyap}.
\end{revision}
\end{proof}

\begin{theorem}\label{thm:nonconvex}
\begin{revision}
Suppose Assumptions~\ref{ass:regularity}--\ref{ass:oracle} hold,
$p/n_c\le\kappa_c$, $W_0\le C_Wn$, and
\begin{align*}
0&<\alpha<\frac{2}{\lambda_n(L)},\\
\eta_k&=\eta=\sqrt{\frac{n}{pT}},\\
0&\le\beta_k\le\bar\beta_k,\qquad
\delta_{i,k}&\le
\frac{\kappa_\delta}
{p^{3/4}n^{1/4}(k+1)^{1/4}}.
\end{align*}
If $T\ge n^3/(p\bar\eta^2)$, then
\begin{align}
\frac1T\sum_{k=0}^{T-1}
\EE[\norm{\nabla f(\bar x_k)}^2]
&=\mathcal O\!\left(\sqrt{\frac{p}{nT}}\right),
\label{eq:nonconvex_rate}\\
\frac1T\sum_{k=0}^{T-1}
\EE\!\left[\frac1n\sum_i\norm{x_{i,k}-\bar x_k}^2\right]
&=\mathcal O\!\left(\frac{n}{T}\right).
\label{eq:consensus_rate}
\end{align}
The hidden constants may depend on the quantities listed in
Lemma~\ref{lem:key}, but not on $T$. In particular, the result does not
require alignment of $n^{-1}\sum_i\mathcal P_{\gamma,\tau}(g^e_{i,k})$.
\end{revision}
\end{theorem}

\begin{proof}
\begin{revision}
Summing~\eqref{eq:lyap} gives
\begin{align}
c_x\sum_{k<T}\EE[\chi_k]
&+c_gn\eta\sum_{k<T}\EE[\norm{\nabla f(\bar x_k)}^2]\notag\\
&\le W_0+C_0pT\eta^2\notag\\
&\quad+C_\delta np^2\eta\sum_{k<T}\delta_k^2.
\label{eq:sum_bound}
\end{align}
The radius schedule implies
$\sum_{k<T}\delta_k^2\le
C\sqrt T/(p^{3/2}\sqrt n)$. Dividing~\eqref{eq:sum_bound} by
$nT\eta$ and using $W_0\le C_Wn$ proves~\eqref{eq:nonconvex_rate}.
Dividing its first term by $T$ proves~\eqref{eq:consensus_rate}.
The condition on $T$ is exactly the requirement
$n\eta\le\bar\eta$ used to absorb disagreement in Lemma~\ref{lem:key}.
\end{revision}
\end{proof}

\begin{assumption}\label{ass:pl}
\rev{The global objective satisfies the Polyak--\L{}ojasiewicz condition
$\frac12\norm{\nabla f(x)}^2\ge\nu(f(x)-f^*)$ for some $\nu>0$.
This is a nonconvex gradient-dominance condition and does not assert convexity.}
\end{assumption}

\begin{theorem}\label{thm:pl}
\begin{revision}
Suppose Assumptions~\ref{ass:regularity}--\ref{ass:pl} hold and
$p/n_c\le\kappa_c$. Choose
\[
\eta_k=\frac{\kappa_\eta}{k+t_0},\qquad
0\le\beta_k\le\bar\beta_k,\qquad
\delta_{i,k}\le\kappa_\delta
\sqrt{\frac{\eta_k}{np}},
\]
where $\kappa_\eta$ is chosen so that
$a:=c_V\kappa_\eta>1$ for the Lyapunov contraction below, and $t_0>a$ is
large enough that $\eta_k\le\bar\eta/n$. Then
\begin{equation}\label{eq:pl_finite}
 \EE[f(\bar x_T)-f^*]
 \le \frac{C}{n}\left(\frac{t_0}{T+t_0}\right)^a W_0
 +\frac{Cp}{n(T+t_0)}.
\end{equation}
Moreover, for some $\rho\in(0,1)$,
\begin{equation}\label{eq:pl_consensus_finite}
 \EE[\chi_T]\le C\rho^T\chi_0+
 \frac{Cp}{(T+t_0)^2}+
 \frac{CW_0t_0^a}{n(T+t_0)^{a+2}}.
\end{equation}
Thus the post-transient objective term is $\mathcal O(p/(nT))$.
\end{revision}
\end{theorem}

\begin{proof}
\begin{revision}
The PL inequality and~\eqref{eq:lyap} imply, for a constant $c_V>0$,
\[
\EE[W_{k+1}]
\le(1-c_V\eta_k)\EE[W_k]+C_0p\eta_k^2+
C_\delta np^2\eta_k\delta_k^2.
\]
The chosen radius makes the last term at most $C p\eta_k^2$; the
nonlinear residual has the same order by~\eqref{eq:beta_schedule}. Thus
\[
 u_{k+1}\le\left(1-\frac{a}{k+t_0}\right)u_k
 +\frac{Cp}{(k+t_0)^2},\qquad u_k=\EE[W_k].
\]
The product estimate in Appendix~D gives
\[
 u_T\le C\left(\frac{t_0}{T+t_0}\right)^a u_0
 +\frac{Cp}{T+t_0},
\]
which proves~\eqref{eq:pl_finite}. Finally,~\eqref{eq:consensus_step},
$\norm{\nabla f(x)}^2\le2L_f(f(x)-f^*)$, and this finite-time bound
give a geometrically stable recursion whose forcing is
$Cp/(k+t_0)^2+C W_0t_0^a/[n(k+t_0)^{a+2}]$.
Its convolution proves~\eqref{eq:pl_consensus_finite}.
\end{revision}
\end{proof}

\begin{remark}\label{rem:endpoints}
\rev{For $\gamma=1$ or $\beta_k=0$, $r_{i,k}=0$ and the analysis reduces
to the usual ZOOM proof. For $\gamma<1$,
schedule~\eqref{eq:beta_schedule} controls the possibly misaligned nonlinear
residual.}
\begin{revision}
The pair $(\gamma,\tau)$ sets the nonlinear response. Smaller $\gamma$
increases both weak-component actuation and sensitivity to noise near
zero; $\tau$ records the normalization at which amplification changes
to attenuation. The weight $\beta_k$ controls departure from the raw
ZO direction. Equation~\eqref{eq:beta_schedule} gives the admissible range
$\beta_k\in[0,\bar\beta_k]$. The condition $16R_g\beta_{\rm d}^2\le1$
makes the worst-case residual subordinate to raw descent; the
stepsize-dependent cap preserves the rate order.
None of these parameters changes the query or message
budget, and $\beta_k$ requires no stored vector state. They should be
reported with the objective scaling and selected by a validation or
noise-sensitivity study. The radius schedules in the
theorems correspond to the common-random-number oracle; independent
measurement noise requires the modification in Remark~\ref{rem:indnoise}.
\end{revision}
\end{remark}

\section{Numerical Examples}

\subsection{Black-Box Binary Classification}

\begin{revision}
We use the nonlinear least-squares benchmark in~\cite{liu2018zeroth},
$f_i(x)=(y_i-\phi(x;a_i))^2$, where
$\phi(x;a)=1/(1+\exp(-a^\top x))$. The synthetic covariates are drawn
from $\mathcal N(0,I_p)$. With $x_\star=\one_p/\sqrt p$, labels are
$y=\bm{1}\{a^\top x_\star+\epsilon\ge0\}$, where
$\epsilon\sim\mathcal N(0,0.1^2)$; the training and test sets contain
$2000$ and $200$ samples.
We set $p=100$, $n=10$, $n_c=10$, batch size one, $T=10^4$, and use a
connected Erd\H{o}s--R\'enyi graph. All eight methods use the same data,
graph, seed protocol, 20 function queries per agent per round, and one
broadcast primal-size vector. The coordinate family uses the two-sided
estimator~\eqref{eq:est2} with $n_c=10$. ZOD-PDA and ZOD-PA average ten
one-sided spherical probes, ZO-GDA averages ten symmetric spherical probes,
and ZONE-M uses $J=10$ Gaussian probes. The ten-probe averages are explicit
budget-matching minibatch adaptations of the cited two-point recursions;
their published guarantees retain their original oracle assumptions, and
ZO-GDA's cited analysis is deterministic.  This batching is needed to match
both budgets: running ten original two-call rounds would spend the same
queries but ten primal messages.

ZOOM is the raw-ZO ablation of ZOOM-PB and uses the same $\alpha=0.035$,
$\eta_k=0.08/(k+25)^{0.05}$, and $\delta_k=0.08/(k+1)^{0.20}$.
For ZOOM-PB, $\gamma=0.7$, $\tau=5$, and
$\beta_k=\min\{0.65,8\sqrt{\eta_k/n}\}$.
The tuned ZODIAC and ZODIAC-PB parameters are
$(\eta,\alpha,\beta)=(0.08,0.25,0.1875)$ and $(0.05,0.4,0.4)$.
ZOD-PDA uses $(\eta,\alpha,\beta)=(0.15,0.15,0.10)$.
ZOD-PA uses consensus gain $0.02$ and
$\eta_k=0.18/(k+25)^{0.05}$; ZO-GDA uses
$\eta_k=0.15/(k+25)^{0.05}$; and ZONE-M uses $\rho=1.5$.
All non-ZOOM parameters were selected using pilot seeds 100 and 101, whereas
the displayed means use the disjoint seeds $0,\ldots,9$.
\end{revision}

\begin{figure}[H]
\centering
\includegraphics[height=1.90in,keepaspectratio]{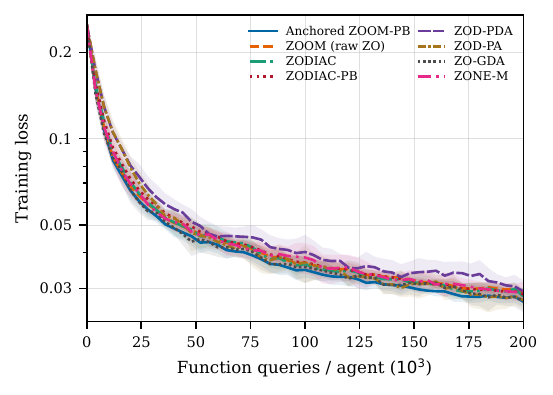}
\caption{\rev{Matched black-box comparison on a log scale (mean $\pm$
s.d., ten new seeds). The horizontal axis counts per-agent queries. Every
method uses 20 queries and broadcasts one primal-size vector per round;
ZOD-PDA, ZODIAC, and ZODIAC-PB additionally maintain a local dual vector.}}
\label{fig:classification}
\end{figure}

\begin{figure*}[t]
\centering
\includegraphics[width=\textwidth]{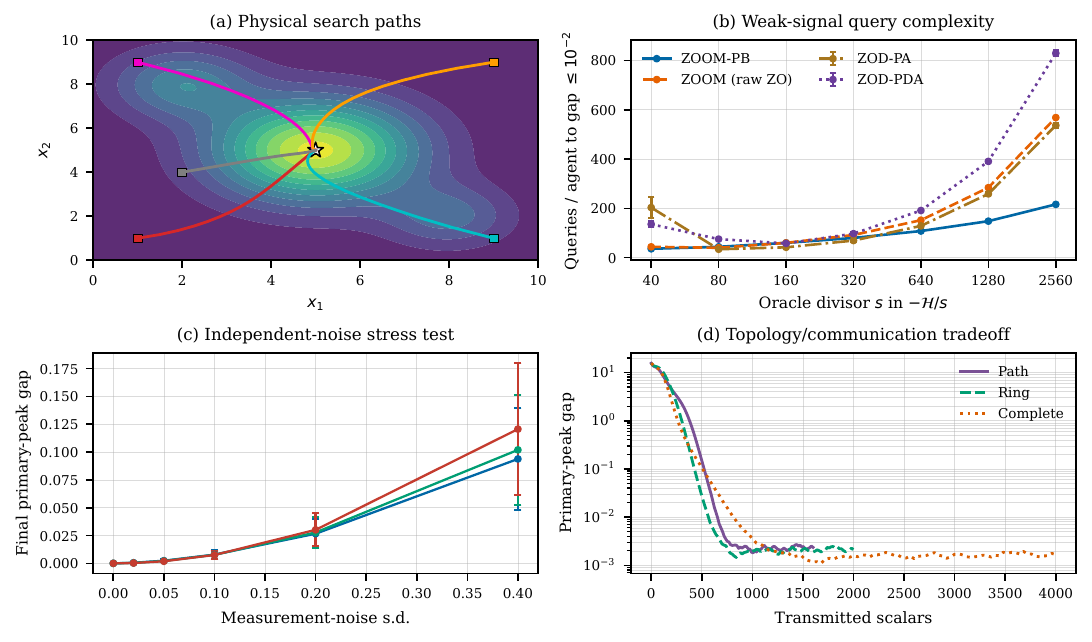}
\caption{\rev{UAV source seeking with $\beta_k=\sqrt{\eta_k/n}$. (a)
Physical query paths for $\gamma=0.7$ without measurement noise;
squares/circles mark initial/final positions. (b) Per-agent queries to gap
$10^{-2}$ as the divisor $s$ in $-\mathcal H/s$ weakens the signal; all
four methods use four calls and one sent vector per round. (c) Final gap
versus independent measurement noise. (d) Topologies versus transmitted
scalars. Error bars show 30-seed s.d. where randomness is present.}}
\label{fig:uav_robustness}
\end{figure*}

\begin{revision}
The terminal losses are $0.0270\pm0.0010$ (ZOOM-PB),
$0.0288\pm0.0011$ (ZOOM), $0.0282\pm0.0014$ (ZODIAC),
$0.0286\pm0.0010$ (ZODIAC-PB), $0.0293\pm0.0027$ (ZOD-PDA),
$0.0273\pm0.0020$ (ZOD-PA), $0.0272\pm0.0024$ (ZO-GDA), and
$0.0293\pm0.0023$ (ZONE-M). Mean test accuracies range from $95.4\%$ to
$96.15\%$. Thus ZOOM-PB has the lowest mean terminal loss in this matched
experiment and uses one fewer local $p$-vector than ZOD-PDA. The overlapping
uncertainty bands and the primal-only ZOD-PA comparison preclude a claim of
uniform dominance, a better asymptotic order, or superior generalization.
\end{revision}

\subsection{UAV Source Seeking}

\begin{revision}
We next use the iterates as physical search/query waypoints for five
UAVs. At this kinematic planning layer, low-level tracking dynamics are
not modeled. The unknown concentration field is
\begin{equation}
\mathcal H(x)=\sum_{j=1}^{3}A_j
\exp\!\left(-\frac{\norm{x-c_j}^2}{2\sigma_j^2}\right),
\end{equation}
with primary peak $c_1=[5,5]^\top$ and two weaker peaks. Agent $i$
observes $Y_i(x)=\mathcal H(x)+\omega_i$ and queries
$F_i(x,\omega_i)=-Y_i(x)/40$; it has no field map or gradient. Each
round uses the two-sided oracle with $n_c=p=2$ and exchanges only
$x_{i,k}$. We set a ring graph, $\alpha=0.055$,
$\eta_k=1.10/(k+1)^{0.12}$, $\delta_k=0.14/(k+1)^{0.20}$,
$\tau=0.05$, and $\beta_k=\sqrt{\eta_k/n}$.
\end{revision}

\begin{revision}
The plotted metric is
$\mathcal H(c_1)-n^{-1}\sum_i\mathcal H(x_{i,k})$. Figure~\ref{fig:uav_robustness}(a)
therefore shows physical query paths, not latent source estimates.
Panel (b) isolates useful weak signals without measurement noise by sweeping
$s\in\{40,80,160,320,640,1280,2560\}$ in $-\mathcal H/s$. The final two
scales were not used for tuning. At $s=1280$, ZOOM-PB reaches gap $10^{-2}$
in 148 queries per agent, versus 284 for ZOOM, $259\pm8$ for ZOD-PA, and
$390\pm7$ for ZOD-PDA. At $s=2560$, the corresponding counts are 216,
568, $536\pm10$, and $829\pm12$. Thus the reductions across the two
held-out scales range from $43\%$ to $74\%$. Every method uses one fixed
pilot-selected tuning across all scales, four calls, and one sent vector
per round; the dual state of ZOD-PDA remains local. This result supports a
finite-time weak-signal advantage, not a better asymptotic order.
Panel (c) exposes the complementary limitation: as independent measurement
noise increases, smaller $\gamma$ attenuates the resulting large
finite-difference inputs. At noise s.d. $0.05$, $\gamma=1$ has the smallest
mean gap; at s.d. $0.4$, $\gamma=0.5$ lowers it by about $22\%$. This stress
test is outside the theorem oracle. Panel (d) shows that the complete graph
improves the terminal gap but uses twice as many transmitted scalars as the
ring. Relative to the path, it lowers the 100-round gap by $29\%$ and
$\chi_{100}$ by $40\%$, but transmits $2.5$ times as many scalars. Thus denser
connectivity is a communication--performance tradeoff rather than free
acceleration; oracle calls and local state are unchanged.
Table~\ref{tab:uav_endpoints} gives the endpoint statistics.
\end{revision}

\begin{table}[H]
\caption{UAV endpoint statistics over 30 runs (mean $\pm$ s.d.).}
\label{tab:uav_endpoints}
\centering
\scriptsize
\setlength{\tabcolsep}{3pt}
\begin{tabular}{@{}lccc@{}}
\toprule
\multicolumn{4}{c}{Ring: final gap $(\times10^{-2})$}\\
$\gamma$ & noise $0.05$ & noise $0.20$ & noise $0.40$ \\
\midrule
$0.5$ & $0.240\pm0.123$ & $2.662\pm1.319$ & $9.379\pm4.564$ \\
$0.7$ & $0.213\pm0.108$ & $2.774\pm1.371$ & $10.194\pm4.972$ \\
$1.0$ & $0.187\pm0.093$ & $3.011\pm1.486$ & $12.062\pm5.918$ \\
\midrule
\multicolumn{4}{c}{$\gamma=0.7$, noise $0.05$, 100 rounds}\\
Graph & scalars & gap $(\times10^{-3})$ & $\chi_{100}(\times10^{-4})$ \\
\midrule
Path & $1600$ & $2.361\pm1.113$ & $5.25$ \\
Ring & $2000$ & $2.135\pm1.082$ & $4.57$ \\
Complete & $4000$ & $1.673\pm0.870$ & $3.17$ \\
\bottomrule
\end{tabular}
\end{table}

All repeated UAV tests use the same five initial waypoints, seeds
$0,\ldots,29$, and projection onto $[0,10]^2$. The weak-signal panel uses
$T=300$; the other repeated panels use $T=100$. Its method-specific parameters
were selected on direction seeds 100--101 at $s\in\{40,160,640\}$ and then
fixed across the seven displayed scales. In the noise and topology tests, the
positive and negative function evaluations draw independent Gaussian noises.
The noise sweep changes only $(\gamma,\sigma_\omega)$; the topology sweep fixes
$\gamma=0.7$, $\sigma_\omega=0.05$, and all update parameters. Its horizontal
budget counts $2|\mathcal E|p$ transmitted scalars per completed round. Thus
neither comparison adds oracle calls or communicated states to the nonlinear
variant.

\FloatBarrier

\section{Conclusion}

\begin{revision}
ZOOM-PB combines coordinate function queries, a componentwise nonlinear
gain, and a one-state recursion. Retaining the raw ZO direction as a
linear anchor avoids an aggregate-alignment assumption: the possibly
misaligned nonlinear component is controlled as a vanishing perturbation.
The resulting method has the stated nonconvex order and PL statistical
term without an auxiliary tracking state. The broad classification comparison
gives comparable endpoints. In the weak-signal UAV sweep, ZOOM-PB uses
$43$--$74\%$ fewer queries at the two held-out weakest scales than ZOOM,
ZOD-PA, and ZOD-PDA under fixed per-method tuning; the noise sweep shows why this is not
uniform superiority. Extending the theory to
independent measurement noise and adaptive selection of
$(\gamma,\tau,\beta_k)$ remains open.
\end{revision}

\section{Appendix A: Raw-Oracle and Nonlinear-Residual Bounds}

Let
\[
m_{i,k}:=\EE[g^e_{i,k}\mid\mathcal F_k].
\]
For the two-sided estimator, define
\begin{align}
d_{i,\ell,k}
&:=\frac{F_i(x_{i,k}+\delta_{i,k}e_\ell,\xi_{i,k})
-F_i(x_{i,k}-\delta_{i,k}e_\ell,\xi_{i,k})}{2\delta_{i,k}},\notag\\
g^e_{i,k}
&=\frac{p}{n_c}\sum_{\ell\in\mathcal S_{i,k}}
d_{i,\ell,k}e_\ell.                                               \label{eq:app_estimator_expanded}
\end{align}
Under the common-random-number model in
Assumptions~\ref{ass:smooth}--\ref{ass:oracle}, uniform sampling without
replacement gives the exact conditional identities
\begin{align}
\EE_{\mathcal S_{i,k}}[g^e_{i,k}\mid\xi_{i,k},\mathcal F_k]
&=\sum_{\ell=1}^p d_{i,\ell,k}e_\ell,\notag\\
\EE_{\mathcal S_{i,k}}[\norm{g^e_{i,k}}^2\mid\xi_{i,k},\mathcal F_k]
&=\frac{p}{n_c}\sum_{\ell=1}^p d_{i,\ell,k}^2.                    \label{eq:app_sampling_identities}
\end{align}
Moreover, the fundamental theorem of calculus and sample-function
smoothness give
\begin{align}
\left|d_{i,\ell,k}
-\partial_\ell F_i(x_{i,k},\xi_{i,k})\right|
&\le\frac{1}{2\delta_{i,k}}
\int_{-\delta_{i,k}}^{\delta_{i,k}}L_f|t|\,dt\notag\\
&=\frac{L_f\delta_{i,k}}{2}.                                      \label{eq:app_coordinate_remainder}
\end{align}
The smoothness and finite second-moment assumptions justify differentiation
under the expectation, so
$\EE[\nabla F_i(x,\xi_i)]=\nabla f_i(x)$. Taking expectations over $\xi_{i,k}$ in
\eqref{eq:app_sampling_identities}, then using the coordinate-noise bound
and~\eqref{eq:app_coordinate_remainder}, yields
\begin{align}
\norm{m_{i,k}-\nabla f_i(x_{i,k})}
&\le C\sqrt p\,\delta_k,\label{eq:app_raw_bias}\\
\EE[\norm{g^e_{i,k}-m_{i,k}}^2\mid\mathcal F_k]
&\le C\frac{p}{n_c}
\left(\norm{\nabla f_i(x_{i,k})}^2+p\zeta^2\right)\notag\\
&\quad+C\frac{p^2}{n_c}\delta_k^2.
\label{eq:app_raw_variance}
\end{align}
The second line separates coordinate-subsampling variation from the
sample-function variation.  The one-sided estimator follows by replacing
the symmetric remainder with its one-sided counterpart and changing only
the numerical constants.

These estimates use the same $\xi_{i,k}$ at the two perturbed points.
If independent additive noises $w^+$ and $w^-$, each with variance
$\nu_i^2$, are used instead, the symmetric coordinate difference has
noise variance $\nu_i^2/(2\delta_{i,k}^2)$ per selected coordinate.
After multiplication by $p/n_c$, its contribution to the estimator
second moment is
\begin{equation}
\frac{p^2\nu_i^2}{2n_c\delta_{i,k}^2}.
\label{eq:app_independent_noise}
\end{equation}
The one-sided estimator has the same $1/\delta_{i,k}^2$ dependence with
a different numerical constant. This proves the distinction stated in
Remark~\ref{rem:indnoise}.

We next derive all moment bounds from the raw oracle. Smoothness and
heterogeneity imply
\begin{align}
\frac1n\sum_i\norm{\nabla f_i(x_{i,k})}^2
&\le C\norm{\nabla f(\bar x_k)}^2+C\chi_k+C\varsigma^2.
\label{eq:app_local_gradients}
\end{align}
Combining~\eqref{eq:app_raw_bias}--\eqref{eq:app_raw_variance} with
\eqref{eq:app_local_gradients}, and using $p/n_c\le\kappa_c$, yields
\begin{equation}
 \frac1n\sum_i\EE_k\norm{g^e_{i,k}}^2
 \le C\big(\norm{\nabla f(\bar x_k)}^2+\chi_k+p+p^2\delta_k^2\big).
\label{eq:app_raw_local}
\end{equation}
The constant absorbs $\zeta$, $\varsigma$, and $\kappa_c$, but is
independent of $n,p,T$.

Because $n^{-1}\sum_i\nabla f_i(\bar x_k)=\nabla f(\bar x_k)$,
\begin{align}
\left\|\bar m_k-\nabla f(\bar x_k)\right\|
&\le \frac1n\sum_i\norm{m_{i,k}-\nabla f_i(x_{i,k})}\notag\\
&\quad+\frac1n\sum_i
\norm{\nabla f_i(x_{i,k})-\nabla f_i(\bar x_k)}\notag\\
&\le C\sqrt p\,\delta_k+L_f\sqrt{\chi_k},
\label{eq:app_average_bias}
\end{align}
where $\bar m_k=n^{-1}\sum_i m_{i,k}$. Let
$z_{i,k}=g^e_{i,k}-m_{i,k}$. The $z_{i,k}$ are conditionally centered
and independent, hence
\begin{align}
\EE_k\norm{\bar g^e_k}^2
&=\norm{\bar m_k}^2+
\frac1{n^2}\sum_i\EE_k\norm{z_{i,k}}^2\notag\\
&\le C\norm{\nabla f(\bar x_k)}^2+C\chi_k
+C\frac pn+Cp^2\delta_k^2.
\label{eq:app_raw_average}
\end{align}

For $a\in\RR$ and $\gamma\in[1/2,1]$,
\begin{align}
|\mathcal P_{\gamma,\tau}(a)-a|^2
&\le2\tau^{2(1-\gamma)}|a|^{2\gamma}+2a^2\notag\\
&\le C_{\gamma,\tau}(1+a^2).
\label{eq:app_scalar_residual}
\end{align}
Consequently,
\[
 \norm{r_{i,k}}^2\le C_{\gamma,\tau}
 \big(p+\norm{g^e_{i,k}}^2\big).
\]
Together with~\eqref{eq:app_raw_local}, this proves
\eqref{eq:residual_moment}. Jensen's inequality gives
\begin{align}
\EE_k\norm{\bar r_k}^2
&\le\frac1n\sum_i\EE_k\norm{r_{i,k}}^2\notag\\
&\le C\big(\norm{\nabla f(\bar x_k)}^2+\chi_k+p+p^2\delta_k^2\big).
\label{eq:app_residual_average}
\end{align}
Since $\bar s_k=\bar g^e_k+\beta_k\bar r_k$,
\eqref{eq:app_raw_average} and~\eqref{eq:app_residual_average} prove
\eqref{eq:moment}. Finally,
\begin{align}
\EE_k\norm{\bsK\bm s_k}^2
&\le\sum_i\EE_k\norm{s_{i,k}}^2\notag\\
&\le2\sum_i\EE_k\norm{g^e_{i,k}}^2
+2\beta_k^2\sum_i\EE_k\norm{r_{i,k}}^2,
\end{align}
which proves~\eqref{eq:centered_moment}. This is the separate network
forcing estimate required by the disagreement recursion.

\section{Appendix B: Detailed Lyapunov Proof}

Let $A_\alpha=I_{np}-\alpha\bsL$. On the disagreement subspace,
\[
\norm{\bsK A_\alpha z}\le\rho_\alpha\norm{\bsK z},
\qquad
\rho_\alpha=\max_{2\le i\le n}|1-\alpha\lambda_i(L)|<1.
\]
Choose $\epsilon_\alpha>0$ so that
$(1+\epsilon_\alpha)\rho_\alpha^2=1-\ell_\alpha$ for some
$\ell_\alpha\in(0,1)$. Applying
$\norm{a+b}^2\le(1+\epsilon_\alpha)\norm a^2+
(1+\epsilon_\alpha^{-1})\norm b^2$ to
\eqref{eq:disagreement_recursion}, dividing by $n$, and using
\eqref{eq:centered_moment} gives
\begin{align}
\EE[\chi_{k+1}\mid\mathcal F_k]
&\le(1-\ell_\alpha)\chi_k\notag\\
&\quad+C_\alpha\eta_k^2\big(
V_g\norm{\nabla f(\bar x_k)}^2+V_x\chi_k\notag\\
&\hspace{31mm}+V_0p+V_\delta p^2\delta_k^2\big).
\label{eq:app_consensus}
\end{align}
Reducing the fixed upper bound on $\eta_k$ if necessary absorbs the
$V_x\eta_k^2\chi_k$ term into $\ell_\alpha\chi_k/4$.

The average recursion and $L_f$-smoothness give
\begin{align}
n\EE[f(\bar x_{k+1})-f^*\mid\mathcal F_k]
&\le n(f(\bar x_k)-f^*)\notag\\
&\quad-n\eta_k\langle\nabla f(\bar x_k),
\EE_k[\bar s_k]\rangle\notag\\
&\quad+\frac{nL_f\eta_k^2}{2}
\EE[\norm{\bar s_k}^2\mid\mathcal F_k].
\label{eq:app_average_descent}
\end{align}
We now establish the descent term without assuming that the pure
powerball average is aligned.  Let
$g_k=\nabla f(\bar x_k)$ and $d_k=\bar m_k-g_k$.
Equation~\eqref{eq:app_average_bias} and Young's inequality imply
\begin{equation}
 \langle g_k,\bar m_k\rangle
 \ge \frac34\norm{g_k}^2-C\chi_k-Cp\delta_k^2.
 \label{eq:app_raw_alignment}
\end{equation}
For the nonlinear residual, no sign is used:
\begin{align}
 \beta_k\big|\langle g_k,\EE_k[\bar r_k]\rangle\big|
 &\le\frac18\norm{g_k}^2
 +2\beta_k^2\EE_k\norm{\bar r_k}^2.            \label{eq:app_residual_young}
\end{align}
Substituting~\eqref{eq:app_residual_average} into
\eqref{eq:app_residual_young} and using the coefficient $R_g$ from
\eqref{eq:residual_moment} show that the additional gradient coefficient
is at most $2R_g\beta_k^2\le1/8$ under~\eqref{eq:beta_schedule}. Since
$\EE_k[\bar s_k]=\bar m_k+\beta_k\EE_k[\bar r_k]$, we obtain
\begin{equation}
 \langle g_k,\EE_k[\bar s_k]\rangle
 \ge\frac12\norm{g_k}^2-C\chi_k-Cp^2\delta_k^2-Cp\beta_k^2.
 \label{eq:app_anchored_alignment}
\end{equation}
The deliberately looser $p^2\delta_k^2$ term covers both estimators
with one expression.

Substituting~\eqref{eq:app_anchored_alignment} and~\eqref{eq:moment}
into~\eqref{eq:app_average_descent}, and choosing $\eta_k$ small enough
to absorb the $B_g\eta_k^2$ term, yields
\begin{align}
n\EE[f(\bar x_{k+1})-f^*\mid\mathcal F_k]
&\le n(f(\bar x_k)-f^*)\notag\\
&\quad-\frac14n\eta_k
\norm{\nabla f(\bar x_k)}^2\notag\\
&\quad+C_an\eta_k\chi_k+C_0p\eta_k^2\notag\\
&\quad+C_\delta np^2\eta_k\delta_k^2.
\label{eq:app_descent}
\end{align}
Here $n\eta_kp\beta_k^2\le\kappa_\beta^2 p\eta_k^2$; the terms with
$\eta_k^2\chi_k$ and $\eta_k^2p^2\delta_k^2$ have been enlarged into
the displayed terms using $\eta_k\le1$.

Define
\[
W_k=n(f(\bar x_k)-f^*)+\lambda\chi_k
\]
with any fixed $\lambda>0$. Multiply~\eqref{eq:app_consensus} by
$\lambda$ and add~\eqref{eq:app_descent}. If
\[
n\eta_k\le
\bar\eta\le\frac{\lambda\ell_\alpha}{4C_a},
\]
the positive $C_an\eta_k\chi_k$ term is absorbed by the consensus
contraction. A further fixed reduction of $\bar\eta$ absorbs the
$\eta_k^2\norm{\nabla f(\bar x_k)}^2$ term. Taking total expectations
then gives~\eqref{eq:lyap}, with
$c_x=\lambda\ell_\alpha/2$ and
$c_g>0$ numerical. This also shows explicitly why the stepsize
restriction is on $n\eta_k$.  The admissible $\bar\eta$ deteriorates as
the residual constant $C_{\gamma,\tau}$ grows or as
$1-\rho_\alpha$ approaches zero.

\section{Appendix C: Nonconvex Rate}

For a constant $\eta$, summing~\eqref{eq:lyap} from $k=0$ to $T-1$
gives
\begin{align}
c_x\sum_{k<T}\EE[\chi_k]
&+c_gn\eta\sum_{k<T}\EE[\norm{\nabla f(\bar x_k)}^2]\notag\\
&\le W_0+C_0pT\eta^2\notag\\
&\quad+C_\delta np^2\eta\sum_{k<T}\delta_k^2.
\label{eq:app_sum}
\end{align}
Consequently,
\begin{align}
\frac1T\sum_{k<T}\EE[\norm{\nabla f(\bar x_k)}^2]
&\le
\frac{W_0}{c_gnT\eta}+
\frac{C_0p\eta}{c_gn}+
\frac{C_\delta p^2}{c_gT}\sum_{k<T}\delta_k^2,
\label{eq:app_stationarity_finite}\\
\frac1T\sum_{k<T}\EE[\chi_k]
&\le
\frac{W_0}{c_xT}+
\frac{C_0p\eta^2}{c_x}+
\frac{C_\delta np^2\eta}{c_xT}\sum_{k<T}\delta_k^2.
\label{eq:app_consensus_finite}
\end{align}
For
\[
\eta=\sqrt{\frac{n}{pT}},\qquad
\delta_{i,k}\le
\frac{\kappa_\delta}
{p^{3/4}n^{1/4}(k+1)^{1/4}},
\]
the elementary estimate
\[
\sum_{k=0}^{T-1}\delta_k^2
\le\frac{2\kappa_\delta^2\sqrt T}{p^{3/2}\sqrt n}
\]
holds. Substitution into~\eqref{eq:app_stationarity_finite} and
\eqref{eq:app_consensus_finite}, together with $W_0\le C_Wn$, yields
\begin{align*}
\frac1T\sum_{k<T}\EE[\norm{\nabla f(\bar x_k)}^2]
&\le C\sqrt{\frac{p}{nT}},\\
\frac1T\sum_{k<T}\EE[\chi_k]&\le C\frac nT.
\end{align*}
Finally, $T\ge n^3/(p\bar\eta^2)$ is equivalent to
$n\eta\le\bar\eta$, the restriction used in Appendix~B.

\section{Appendix D: PL Rate}

Under Assumption~\ref{ass:pl},
\[
\norm{\nabla f(\bar x_k)}^2
\ge2\nu(f(\bar x_k)-f^*).
\]
The consensus decrease in~\eqref{eq:lyap} dominates
$c\eta_k\chi_k$ for all sufficiently small $\eta_k$. Hence there is a
$c_V>0$, depending on $\nu$ and the graph constants, such
that
\begin{align}
\EE[W_{k+1}]
&\le(1-c_V\eta_k)\EE[W_k]+C_0p\eta_k^2\notag\\
&\quad+C_\delta np^2\eta_k\delta_k^2.
\label{eq:app_pl_recursion}
\end{align}
The choice
$\delta_{i,k}^2\le\kappa_\delta^2\eta_k/(np)$ makes the last term at
most $C p\eta_k^2$. Set
$\eta_k=\kappa_\eta/(k+t_0)$,
$a=c_V\kappa_\eta>1$, and take $t_0$ large enough that
$n\eta_k\le\bar\eta$ and $a/t_0<1$.  With
$u_k=\EE[W_k]$,~\eqref{eq:app_pl_recursion} becomes
\[
 u_{k+1}\le\left(1-\frac{c_V\kappa_\eta}{k+t_0}\right)u_k+
 \frac{Cp\kappa_\eta^2}{(k+t_0)^2}
 .
\]
For $0\le j<T$, the elementary bound $1-z\le e^{-z}$ gives
\begin{align}
 \prod_{k=0}^{T-1}\left(1-\frac{a}{k+t_0}\right)
 &\le\left(\frac{t_0}{T+t_0}\right)^a,\label{eq:app_product_initial}\\
 \prod_{k=j+1}^{T-1}\left(1-\frac{a}{k+t_0}\right)
 &\le C_a\left(\frac{j+t_0}{T+t_0}\right)^a.
 \label{eq:app_product_forcing}
\end{align}
Unrolling the recursion and using
$\sum_{j=0}^{T-1}(j+t_0)^{a-2}\le
C_a(T+t_0)^{a-1}$ therefore yields
\begin{equation}
 \EE[W_T]\le
 C\left(\frac{t_0}{T+t_0}\right)^aW_0
 +\frac{Cp}{T+t_0}.
 \label{eq:app_pl_W}
\end{equation}
Since $W_T\ge n(f(\bar x_T)-f^*)$, division by $n$ proves
\eqref{eq:pl_finite}, including its initialization transient.

For the pointwise disagreement,~\eqref{eq:app_consensus} and the
stepsize restriction imply
\begin{align}
\EE[\chi_{k+1}]
&\le(1-\ell_\alpha/2)\EE[\chi_k]\notag\\
&\quad+C\eta_k^2\big(\EE[\norm{\nabla f(\bar x_k)}^2]\notag\\
&\hspace{31mm}+p+p^2\delta_k^2\big).
\end{align}
For an $L_f$-smooth function bounded below by $f^*$,
$\norm{\nabla f(x)}^2\le2L_f(f(x)-f^*)$.  Combining this inequality,
the radius schedule, and~\eqref{eq:app_pl_W} reduces the preceding
recursion to
\begin{align}
 \EE[\chi_{k+1}]&\le q\EE[\chi_k]
 +\frac{Cp}{(k+t_0)^2}\notag\\
 &\quad+\frac{CW_0t_0^a}{n(k+t_0)^{a+2}},
 \qquad q=1-\ell_\alpha/2.                 \label{eq:app_chi_forcing}
\end{align}
For every $b>0$ and $q\in(0,1)$, splitting the convolution at
$j=\lfloor T/2\rfloor$ gives
\[
 \sum_{j=0}^{T-1}\frac{q^{T-1-j}}{(j+t_0)^b}
 \le\frac{C_{b,q,t_0}}{(T+t_0)^b}.
\]
Unrolling~\eqref{eq:app_chi_forcing} with $b=2$ and $b=a+2$ proves
\eqref{eq:pl_consensus_finite}.

\section{Appendix E: Reproducibility Details}

\begin{figure*}[t]
\centering
\includegraphics[width=\textwidth]{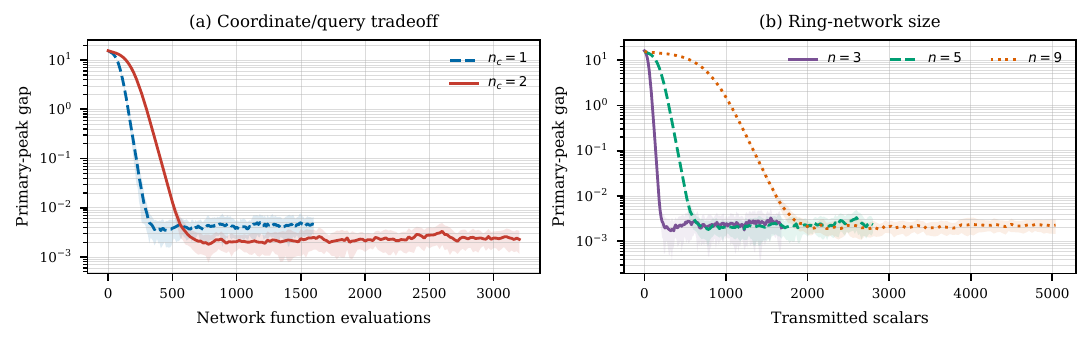}
\caption{Additional UAV scaling tests for $\gamma=0.7$, $\tau=0.05$,
and $\beta_k=\sqrt{\eta_k/n}$. (a) Number of sampled coordinates,
normalized by total network function evaluations. (b) Number of agents
on a ring, normalized by transmitted scalars. Shading is mean $\pm$
s.d. over 30 independent measurement-noise realizations.}
\label{fig:uav_scaling}
\end{figure*}

For classification, covariates are sampled from $\mathcal N(0,I_{100})$.
With $x_\star=\one_{100}/10$, labels are
$\bm{1}\{a^\top x_\star+\epsilon\ge0\}$ for
$\epsilon\sim\mathcal N(0,0.1^2)$. The data set and graph are fixed
across runs; ten oracle seeds independently select local samples,
coordinate subsets, and random directions.  We use batch size one,
$n=10$, $n_c=10$,
$T=10^4$,
\[
 \eta_k=\frac{0.08}{(k+25)^{0.05}},\qquad
 \delta_k=\frac{0.08}{(k+1)^{0.20}},\qquad \alpha=0.035.
\]
Anchored ZOOM-PB uses $\gamma=0.7$, $\tau=5$, and
$\beta_k=\min\{0.65,8\sqrt{\eta_k/n}\}$. ZODIAC uses
$(\eta,\alpha,\beta)=(0.08,0.25,0.1875)$, and ZODIAC-PB uses
$(0.05,0.4,0.4)$. ZOD-PDA uses
$(\eta,\alpha,\beta)=(0.15,0.15,0.10)$. ZOD-PA uses consensus gain
$0.02$ and stepsize scale $0.18$; ZO-GDA uses stepsize scale $0.15$;
ZONE-M uses $\rho=1.5$. These values were selected using pilot oracle
seeds $100$ and $101$; reported results use disjoint seeds $0,\ldots,9$.

Every method makes 20 function evaluations and broadcasts one $p$-vector
per round. Coordinate methods use ten two-sided coordinate differences.
ZOD-PDA and ZOD-PA average ten one-sided spherical differences, ZO-GDA
averages ten symmetric spherical differences, and ZONE-M uses ten Gaussian
differences. The data, graph, sample streams, and method-specific probe
streams are matched seed by seed. Figure~\ref{fig:classification} reports
means and standard deviations.

For UAV source seeking, the field parameters are
\[
(A_1,A_2,A_3)=(17,7,5),\quad
(\sigma_1,\sigma_2,\sigma_3)=(1.6,1.1,1.0),
\]
with centers $(5,5)$, $(2,8)$, and $(8,2)$. The five initial query
positions are $(1,1)$, $(9,1)$, $(1,9)$, $(9,9)$, and $(2,4)$.
The path, noise, and topology panels use
$F_i(x,\omega)=-[\mathcal H(x)+\omega]/40$,
$n_c=p=2$, $\alpha=0.055$, $\tau=0.05$, and
$\beta_k=\sqrt{\eta_k/n}$, with
\[
\eta_k=\frac{1.10}{(k+1)^{0.12}},\qquad
\delta_k=\frac{0.14}{(k+1)^{0.20}}.
\]
The path plot uses a ring, $\gamma=0.7$, zero measurement noise, and
55 rounds. The noise and topology panels use 100 rounds. The noise sweep uses
levels $0$, $0.02$, $0.05$, $0.10$, $0.20$, and $0.40$; the topology sweep
uses $\gamma=0.7$ and noise standard deviation $0.05$. Means and
standard deviations use random seeds $0$ through $29$. Independent
Gaussian noises are drawn separately at the positive and negative
coordinate probes. Consequently, these noisy runs test robustness and
are not instances of the common-random-number oracle assumed in the
rate theorems.

The weak-signal panel uses a ring, no measurement noise, $T=300$, and
$F_i(x)=-\mathcal H(x)/s$ for
$s\in\{40,80,160,320,640,1280,2560\}$. Every round costs four function values per
agent: ZOOM-PB and ZOOM evaluate both coordinates symmetrically, whereas
ZOD-PA and ZOD-PDA use three one-sided spherical probes sharing one center
value. Each method therefore uses four calls and sends
one two-dimensional primal vector; only ZOD-PDA stores a local dual vector.
The stepsize schedule is $c_\eta/(k+1)^{0.12}$. ZOOM-PB uses
$(c_\eta,\alpha,\gamma)=(8,0.055,0.7)$, ZOOM uses $(12,0.055,1)$,
ZOD-PA uses $(12,0.12)$, and ZOD-PDA uses
$(c_\eta,\alpha,\beta)=(8,0.05,0.05)$. One tuning per method was selected by
minimizing the worst pilot query count over $s\in\{40,160,640\}$ using
direction seeds 100 and 101. The candidate $c_\eta$ values were
$\{2,3,4,6,8,12,16,20\}$ for the coordinate methods,
$\{4,6,8,12,16\}$ for ZOD-PA, and $\{2,4,6,8\}$ for ZOD-PDA; their consensus
gains were searched over $\{0.055,0.08,0.12,0.16,0.20\}$ for ZOD-PA and
$(\alpha,\beta)\in\{0.05,0.10,0.20\}^2$ for ZOD-PDA.
The displayed ZOD-PA/ZOD-PDA means and standard
deviations use disjoint direction seeds $0,\ldots,29$. The target is the
unscaled physical gap $10^{-2}$, so changing $s$ weakens only the observed
black-box signal, not the success criterion. The reported $s\in\{1280,2560\}$
points were not included in parameter selection.

The communication axis counts both directions on every undirected
edge. With a two-dimensional primal state, round $k$ corresponds to
$2|\mathcal E|pk$ transmitted scalars. The performance metric is the
primary-peak gap
$\mathcal H(c_1)-n^{-1}\sum_i\mathcal H(x_{i,k})$; the plotted paths
are the physical query waypoints $x_{i,k}$.

Figure~\ref{fig:uav_scaling} reports the two remaining scaling checks.
For both panels, agents start evenly on a radius-$4.3$ circle centered
at $(5,5)$, and the curves show 30-run means and standard deviations at
measurement-noise standard deviation $0.05$. In panel (a), $p=2$ and
$n=5$: sampling one coordinate reaches the source region with fewer
network function evaluations, whereas sampling both coordinates gives
a lower noise floor. In panel (b), $n_c=2$ and $n\in\{3,5,9\}$ on a
ring. Larger rings need more transmitted scalars and have a smaller
spectral gap, so adding agents does not automatically shorten the
communication-normalized transient. This agrees with the explicit
graph dependence of the constants and with the fact that every added
agent also introduces another local objective.

\FloatBarrier

\bibliographystyle{IEEEtran}
\bibliography{bibliography}

\end{document}